\documentclass{article}
\usepackage{spconf,amsmath,amssymb,amsthm,graphicx}
\usepackage{booktabs}
\usepackage{url}
\usepackage{xcolor}
\usepackage{enumitem}
\usepackage{subcaption}
\setlist{nosep, leftmargin=14pt}

\newcommand{\x}{\mathbf{x}}
\newcommand{\y}{\mathbf{y}}
\newcommand{\A}{\mathbf{A}}
\newcommand{\F}{\mathbf{F}}
\newcommand{\D}{\mathbf{D}}
\newcommand{\G}{\mathbf{G}}
\renewcommand{\d}{\mathbf{d}}%
\newcommand{\Dt}{\tilde{\D}}%
\newcommand{\I}{\mathbf{I}}
\renewcommand{\O}{\mathbf{O}}
\renewcommand{\S}{\mathbf{S}}%
\newcommand{\twin}{\mathcal{T}}
\newcommand{\e}{\mathrm{e}}
\renewcommand{\j}{\mathrm{j}}
\newcommand{\conj}[1]{\overline{#1}}
\newcommand{\HH}{^\mathsf{H}}
\newtheorem{proposition}{Proposition}
\newcommand{\ann}[1]{\tag*{\text{(#1)}}}%
\newcommand{\ntv}{\widetilde{\mathrm{TV}}}%
\newcommand{\R}{\mathbb{R}}
\newcommand{\C}{\mathbb{C}}

\title{The Generalized Twin Image: \\Understanding Regularized Phase Retrieval}
\name{Zhiyuan Hu, Michael Unser, Jonathan Dong}
\address{Biomedical Imaging Group, \'Ecole polytechnique f\'ed\'erale de Lausanne, Lausanne, Switzerland}
\begin{document}
\raggedbottom%
\maketitle
\begin{abstract}
Image priors can dramatically improve phase retrieval reconstruction, but their effectiveness strongly varies between imaging modalities. 
We introduce the generalized twin image, which extends the classical twin-image ambiguity from Fourier phase retrieval to any invertible sensing operator, to understand this effectiveness.
The ambiguity extends to undersampled and oversampled measurements, remaining exact under undersampling and for oversampled coded Fourier measurements. 
We show that it provides a simple criterion to assess regularization effectiveness: the image prior should favor the ground truth over its generalized twin image. 
We propose the total-variation ratio as a simple measure of this preference, and derive closed-form estimates for structured random phase retrieval models.
Experiments with a pretrained denoiser demonstrate that the ratio predicts regularized reconstruction performance across different sensing models.

\end{abstract}
\section{Introduction}
\label{sec:intro}

Phase retrieval is a nonlinear inverse problem that arises in a broad range of applications, including crystallography~\cite{harrison1993phase}, holography~\cite{latychevskaia2019tutorial} and ptychography~\cite{maiden2013ptychography}, whenever one aims to reconstruct the underlying complex-valued signal from real-valued intensity measurements. 
Fourier phase retrieval (FPR)~\cite{hayes1982reconstruction}, in which the measurements are the magnitudes of the Fourier transform of the signal, is of central interest and lays the foundation for many other modalities.
The reconstruction difficulty given a modality is closely related to its ambiguities, i.e., false solutions that generate the same measurements as the ground truth.
Besides the global phase shift ambiguity inherent to phase retrieval, FPR has two extra ambiguities: translation and conjugate reflection, the latter also called the \textit{twin image}~\cite{hayes1982reconstruction, bendory2017fourier}. 
In consequence, FPR suffers from the well-known twin-image stagnation problem~\cite{guizar2012twin}, which algorithms impose constraints such as nonnegativity, realness and bounded support~\cite{fienup1982phase} to mitigate.
Uniqueness up to these ambiguities holds under practical conditions~\cite{hayes1982reconstruction}, but no known algorithm is guaranteed to reach the solution~\cite{bendory2017fourier}.

In contrast, dense random models, e.g., an i.i.d. Gaussian matrix, benefit from robust reconstruction guarantees~\cite{candes2015phase}, despite their quadratic complexity in the signal dimension.
Structured random models for phase retrieval have been proposed to combine the efficiency of Fourier-based models and the robustness of random models~\cite{hu2026structured}.
They are a cascade of structured random layers, each of which consists of a Fourier transform and a random diagonal matrix.
The same reconstruction quality as with dense random models can be achieved once there are two layers~\cite{hu2026structured}.

Image priors have greatly enhanced reconstruction quality in inverse problems~\cite{zhang2021plug}.
In phase retrieval, including regularization for structured random models enables reconstruction with one layer, i.e., one random mask added to FPR, and with very few measurements, well below the weak recovery threshold~\cite{hu2026priors, ducotterd2025breaking}.
A central question remains how the effects of regularization depend on the sensing model, and whether they can be quantified.

In this paper, we extend the twin image of FPR to any invertible sensing operator with optional undersampling and oversampling, as the \textit{generalized twin image}.
It is an exact ambiguity with undersampling, and, for a Fourier transform and an optional random mask, also with oversampling by zero-padding.
Feasibility of unregularized reconstruction given sufficient measurements is predicted by whether the generalized twin stays an exact ambiguity.
For regularized reconstruction, image priors become more effective if the generalized twin varies more than the ground truth, as they favor low-variation signals.
We propose the total-variation (TV) ratio to quantify this property, and derive closed-form estimates for structured random models of different depths.
We numerically validate these estimates and perform regularized reconstruction with a pretrained denoiser at different model depths, showing that the TV ratio predicts reconstruction accuracy.
The generalized twin provides a scheme to analyze the optimization landscape for other phase retrieval modalities.

\section{Background}
\label{sec:background}
\subsection{Structured Random Phase Retrieval}
Phase retrieval aims to recover a signal $\x \in \C^n$ from the intensity measurements
\begin{equation}
    \y = |\A \x|^{2},
\end{equation}
where $\A \in \C^{m \times n}$ is the sensing matrix, $|\cdot|$ is the elementwise modulus, and $\y \in \R^m$.
The sampling ratio is $\alpha = m / n$.
We consider sensing matrices of the form
\begin{equation} \label{eq:sgo}
    \A = \S \G \O,
\end{equation}
where $\O \in \R^{p \times n}$ zero-pads the signal to size $p = \max(m, n)$, $\G \in \C^{p \times p}$ is invertible, and $\S \in \R^{m \times p}$ selects $m$ of the $p$ entries.
Hence $\O = \I$ with undersampling, $\alpha \leq 1$, and $\S = \I$ with oversampling, $\alpha > 1$.
Structured random models~\cite{hu2026structured} construct $\G$ as a cascade of Fourier transforms and random diagonal matrices,
\begin{equation} \label{eq:srpr}
    \G = \Bigl[\prod_{i=1}^{N} \F \D_{i} \Bigr]\, \F^{s},
\end{equation}
where $\F \in \C^{p \times p}$ is the unitary discrete Fourier transform (DFT), $\F^{s}$ with $s \in \{0, 1\}$ an optional DFT, and $\D_i \in \C^{p \times p}$ a diagonal matrix with i.i.d. Rademacher entries, i.e., uniform in $\{-1, 1\}$.
For these models, $\O$ pads the signal at the center and $\S$ keeps the $m$ central Fourier coefficients.
We define the \textit{depth} as $L = N + 0.5s$.
When $L = 0.5$, $\G = \F$ gives FPR, and when $L = 1$, $\G = \F \D_1$ gives a single coded diffraction pattern~\cite{candes2015coded}.
Deeper models give $\G = \F \D_1 \F$, $\F \D_1 \F \D_2$ and $\F \D_1 \F \D_2 \F$ for $L = 1.5$, $2$ and $2.5$.

\subsection{The Twin Image Problem}
The classical twin image in FPR~\cite{guizar2012twin} is the conjugate reflection
\begin{equation} \label{eq:twin}
    \twin(\x)[\mathbf{k}] = \conj{\x[-\mathbf{k}]},
\end{equation}
where $\conj{(\cdot)}$ denotes the complex conjugate and $\mathbf{k}$ is counted from half a pixel before the pixel $(\lfloor h/2 \rfloor, \lfloor w/2 \rfloor)$ of an $h \times w$ image, the origin of the centered DFT.
Along each axis, this reverses the pixel order for even sizes, e.g., from $(1, 2, 3, 4)$ to $(4, 3, 2, 1)$, and additionally shifts it by one pixel for odd sizes, e.g., from $(1, 2, 3)$ to $(2, 1, 3)$.
The twin satisfies
\begin{equation} \label{eq:twin-fourier}
    \F \twin(\x) = \boldsymbol{\Lambda}\, \conj{\F \x},
\end{equation}
where $\boldsymbol{\Lambda}$ is diagonal with the entry $\e^{\j 2 \pi (u_1 / h + u_2 / w)}$ at the frequency $(u_1, u_2)$ counted from the zero frequency, so that $|\F \twin(\x)| = |\F \x|$.
Since $\x$ and $\twin(\x)$ produce the same Fourier magnitudes, iterative algorithms often stagnate in a superposition of the two, in which features of both images remain~\cite{guizar2012twin}.
With integer indices $\mathbf{k}$ counted from the origin of the centered DFT, \eqref{eq:twin} reads $\twin(\x)[\mathbf{k}] = \conj{\x[-\mathbf{k} - \mathbf{1}]}$ with periodic indices, where $\mathbf{1} = (1, 1)$.
Substituting $\mathbf{k} \mapsto -\mathbf{k} - \mathbf{1}$ in the DFT sum gives \eqref{eq:twin-fourier} and its counterpart
\begin{equation} \label{eq:twin-fourier-dual}
    \twin(\F \x) = \F \conj{\boldsymbol{\Lambda} \x},
\end{equation}
where $\boldsymbol{\Lambda}$ is applied in the pixel domain with the same indices.
Since $\twin$ acts pixelwise, it also satisfies the product rule
\begin{equation} \label{eq:twin-product}
    \twin(\D \x) = \twin(\D) \twin(\x),
\end{equation}
where $\twin(\D)$ is the diagonal matrix with the diagonal $\twin(\d)$, and $\d$ is the diagonal of $\D$.

\subsection{Reconstruction Algorithms}
Gradient descent with spectral initialization~\cite{candes2015phase} is the standard algorithm for random sensing models.
Spectral methods take as the initial estimate the leading eigenvector of $\A\HH \operatorname{Diag}(\tau(\y)) \A$, usually computed by power iterations, where $\tau$ is an elementwise preprocessing of the measurements~\cite{maillard2022construction}.
The gradient descent update on the amplitude loss $\ell = \frac{1}{2} \|\sqrt{\y} - |\A\x|\|_2^2$~\cite{wang2018solving} is
\begin{equation} \label{eq:gradient}
    \x_+ = \x - \mu \A\HH\operatorname{Diag}\left( \frac{1}{2} - \frac{\sqrt{\y}}{2 |\A \x|} \right) \A \x,
\end{equation}
where $\mu$ is the step size, $\mathsf{H}$ denotes the conjugate transpose and $\operatorname{Diag}(\cdot)$ the diagonal matrix from the given vector.

A regularized reconstruction further interleaves a denoising step $\x_+ = D_\sigma(\x)$ after each gradient step, as in plug-and-play (PnP)~\cite{venkatakrishnan2013plug}, where $D_\sigma$ is a denoiser and $\sigma$ the denoising strength.

\section{Methods}
\label{sec:method}

\subsection{The Generalized Twin Image}
\label{sec:gtwin}
We define the generalized twin image of $\x$ given the sensing matrix $\A = \S \G \O$ as
\begin{equation} \label{eq:gtwin}
    \x' = \O^{\mathsf{T}} \mathbf{v}, \quad \mathbf{v} = \G^{-1} \F \twin(\F\HH \G \O \x),
\end{equation}
where $\mathbf{v}$ is the twin on the padded grid and $\O^{\mathsf{T}}$ crops it to the signal size.
For FPR, $\G = \F$ gives $\x' = \twin(\x)$, and for $L = 1$, the same ambiguity appears in~\cite{fannjiang2012absolute} for phase retrieval with a single random mask.
We show that $\x'$ generates the same measurements as $\x$ whenever $\mathbf{v}$ stays within the support of the ground truth.
\begin{proposition} \label{prop:ambiguity}
If $\mathbf{v} \in \operatorname{range}(\O)$, then $|\A \x'|^{2} = |\A \x|^{2}$.
\end{proposition}
\begin{proof}
Since $\O^{\mathsf{T}} \O = \I$, $\O \O^{\mathsf{T}}$ is the orthogonal projection onto $\operatorname{range}(\O)$, so the condition gives $\O \O^{\mathsf{T}} \mathbf{v} = \mathbf{v}$.
Hence
\begin{align*}
    |\A \x'| &= |\S \G \O \O^{\mathsf{T}} \mathbf{v}| \ann{with \eqref{eq:sgo}, \eqref{eq:gtwin}} \\
    &= |\S \G \mathbf{v}| \ann{using $\O \O^{\mathsf{T}} \mathbf{v} = \mathbf{v}$} \\
    &= |\S \F \twin(\F\HH \G \O \x)| \ann{with \eqref{eq:gtwin}} \\
    &= |\S \boldsymbol{\Lambda} \conj{\G \O \x}| \ann{with \eqref{eq:twin-fourier}} \\
    &= |\S \conj{\G \O \x}| = |\conj{\S \G \O \x}| = |\A \x| . \ann{using $|\boldsymbol{\Lambda}| = \I$, $\conj{\S} = \S$}
\end{align*}
\end{proof}
In particular, the condition holds for every $\x$ with undersampling, where $\O = \I$.
With oversampling, it holds for $L \leq 1$ when the support of $\x$ is symmetric about the reflection center, since $\mathbf{v} = \D_1^{-1} \twin(\D_1) \twin(\O \x)$ only reflects and flips signs.
For $L \geq 1.5$, the inner Fourier transforms spread the twin over the whole padded grid, so oversampling removes the ambiguity.
\subsection{Total-Variation Ratio}
\label{sec:tvratio}
Given an $h \times w$ image $\x$ with $n = h w$, let $\nabla_1 \x, \nabla_2 \x \in \C^{n}$ be its first-order differences along the two axes on the periodic grid, and let $\nabla \x \in \C^{2 n}$ stack the two.
The total variation (TV)~\cite{rudin1992nonlinear} is 
\begin{equation}
    \mathrm{TV}(\x) = \|\nabla \x\|_1.
\end{equation}
We define the normalized TV and its quadratic counterpart as
\begin{equation} \label{eq:tv-g}
    \ntv(\x) = \frac{\|\nabla \x\|_{1}}{2 n} , \qquad G(\x) = \frac{\|\nabla \x\|_{2}}{\sqrt{2 n}}.
\end{equation}
The TV ratio between the generalized twin and the ground truth is
\begin{equation}
    \rho_{\mathrm{TV}} = \frac{\mathrm{TV}(\x')}{\mathrm{TV}(\x)} = \frac{\ntv(\x')}{\ntv(\x)}, \label{eq:rho-tv}
\end{equation}
which quantifies the variation of the generalized twin relative to the ground truth.
A ratio $\rho_{\mathrm{TV}} > 1$ means that the generalized twin varies more than the ground truth, while $\rho_{\mathrm{TV}} = 1$ gives no distinction between the two.
This metric needs only the ground truth and the sensing model, and it does not depend on the image prior.
It therefore serves as a general diagnostic for regularized reconstruction.

\subsection{Closed-Form Estimates of the TV Ratio}
\label{sec:forms}
While $\rho_{\mathrm{TV}}$ can be computed empirically for any image and sensing operator, closed forms reveal which image statistics the ratio depends on, allowing comparison of the models.
We derive closed-form estimates of $\rho_\mathrm{TV}$ for structured random models of different depths with Rademacher masks, assuming that the signal $\x = \e^{\j \boldsymbol{\varphi}}$ has unit magnitude and small phase steps.
We assume undersampling, $\alpha \leq 1$, so that $\O = \I$.
The TV ratio in the oversampled case follows similarly, as the structure of the generalized twin remains.
We denote $\Dt_i = \D_i^{-1} \twin(\D_i) = \D_i \twin(\D_i)$, a diagonal matrix with the entries $\d_i[\mathbf{k}] \, \d_i[-\mathbf{k} - \mathbf{1}] \in \{-1, 1\}$, and $\mathcal{E}$ the set of the $2 n$ pairs $(\mathbf{a}, \mathbf{b})$ of adjacent pixels of the periodic grid.
\begin{proposition} \label{prop:rho-zero}
For $L = 0.5$, the TV ratio is given by
\begin{equation} \label{eq:rho-zero}
    \rho_{\mathrm{TV}} = 1 .
\end{equation}
\end{proposition}
\begin{proof}
\begin{align*}
    \rho_{\mathrm{TV}} &= \frac{\|\nabla \F\HH \F \twin(\F\HH \F \x)\|_{1}}{\|\nabla \x\|_{1}} \ann{with \eqref{eq:srpr}, \eqref{eq:gtwin}, \eqref{eq:rho-tv}, $\G = \F$} \\
    &= \frac{\|\nabla \twin(\x)\|_{1}}{\|\nabla \x\|_{1}} \ann{$\F$ unitary} \\
    &= 1.
\end{align*}
The last equality holds as the twin image keeps the same structure as the ground truth, hence the same TV.
\end{proof}
\begin{proposition} \label{prop:rho-one}
For $L = 1$, the TV ratio is approximated as
\begin{equation} \label{eq:rho-one}
    \rho_{\mathrm{TV}} \approx \frac{1}{2} + \frac{1}{\ntv(\x)} .
\end{equation}
\end{proposition}
\begin{proof}
\begin{align*}
    \rho_{\mathrm{TV}} &= \frac{\|\nabla \D_1^{-1} \F\HH \F \twin(\F\HH \F \D_1 \x)\|_{1}}{\|\nabla \x\|_{1}} \ann{with \eqref{eq:srpr}, \eqref{eq:gtwin}, \eqref{eq:rho-tv}, $\G = \F \D_1$} \\
    &= \frac{\|\nabla \Dt_1 \twin(\x)\|_{1}}{\|\nabla \x\|_{1}} \ann{with \eqref{eq:twin-product}, $\F$ unitary} \\
    &= \frac{\sum_{(\mathbf{a}, \mathbf{b}) \in \mathcal{E}} |\x[\mathbf{b}] - s_{\mathbf{a}\mathbf{b}} \, \x[\mathbf{a}]|}{\|\nabla \x\|_{1}} \\
    &\approx \frac{\tfrac{1}{2} \sum_{(\mathbf{a}, \mathbf{b}) \in \mathcal{E}} |\x[\mathbf{b}] - \x[\mathbf{a}]| + \tfrac{1}{2} \sum_{(\mathbf{a}, \mathbf{b}) \in \mathcal{E}} |\x[\mathbf{b}] + \x[\mathbf{a}]|}{\|\nabla \x\|_{1}} \\
    &= \frac{1}{2} + \frac{\sum_{(\mathbf{a}, \mathbf{b}) \in \mathcal{E}} |\cos((\boldsymbol{\varphi}[\mathbf{b}] - \boldsymbol{\varphi}[\mathbf{a}]) / 2)|}{\|\nabla \x\|_{1}} \ann{half angle} \\
    &\approx \frac{1}{2} + \frac{2 n}{\|\nabla \x\|_{1}} = \frac{1}{2} + \frac{1}{\ntv(\x)} . \ann{with \eqref{eq:tv-g}}
\end{align*}
In the third line, $s_{\mathbf{a}\mathbf{b}} = \Dt_1[\mathbf{a}] \, \Dt_1[\mathbf{b}]$ is a random sign.
The first approximation replaces each term by its expectation over $s_{\mathbf{a}\mathbf{b}}$, which is uniform in $\{-1, 1\}$ unless the rare cases $\mathbf{b} = -\mathbf{a} - \mathbf{1}$, where $s_{\mathbf{a}\mathbf{b}} = 1$.
The second approximation follows as the signal has unit magnitudes and small phase steps.
\end{proof}
\begin{proposition} \label{prop:rho-half}
For $L = 1.5$, the TV ratio is approximated as
\begin{equation} \label{eq:rho-half}
    \rho_{\mathrm{TV}} \approx \frac{\sqrt{\pi}}{2} \, \frac{G(\x)}{\ntv(\x)} .
\end{equation}
\end{proposition}
\begin{proof}
With $\mathbf{g}_i = \nabla_i \conj{\boldsymbol{\Lambda} \x}$,
\begin{align*}
    \rho_{\mathrm{TV}} &= \frac{\sum_{i = 1}^{2} \|\nabla_i \F\HH \D_1^{-1} \F\HH \F \twin(\F\HH \F \D_1 \F \x)\|_{1}}{\|\nabla \x\|_{1}} \ann{with \eqref{eq:srpr}, \eqref{eq:gtwin}, \eqref{eq:rho-tv}, $\G = \F \D_1 \F$} \\
    &= \frac{\sum_{i = 1}^{2} \|\nabla_i \F\HH \Dt_1 \F \conj{\boldsymbol{\Lambda} \x}\|_{1}}{\|\nabla \x\|_{1}} \ann{with \eqref{eq:twin-fourier-dual}, \eqref{eq:twin-product}, $\F$ unitary} \\
    &= \frac{\sum_{i = 1}^{2} \|\F\HH \Dt_1 \F \mathbf{g}_i\|_{1}}{\|\nabla \x\|_{1}} \ann{circulants commute~\cite{gray2006toeplitz}} \\
    &\approx \frac{\sqrt{\pi}}{2} \, \frac{\sqrt{n} \sum_{i = 1}^{2} \|\F\HH \Dt_1 \F \mathbf{g}_i\|_{2}}{\|\nabla \x\|_{1}} \\
    &= \frac{\sqrt{\pi}}{2} \, \frac{\sqrt{n} \sum_{i = 1}^{2} \|\mathbf{g}_i\|_{2}}{\|\nabla \x\|_{1}} \ann{Parseval's theorem, $|\Dt_1| = \I$} \\
    &\approx \frac{\sqrt{\pi}}{2} \, \frac{\sqrt{2 n} \, \|\nabla \conj{\boldsymbol{\Lambda} \x}\|_{2}}{\|\nabla \x\|_{1}} = \frac{\sqrt{\pi}}{2} \, \frac{G(\boldsymbol{\Lambda} \x)}{\ntv(\x)} \ann{with \eqref{eq:tv-g}} \\
    &\approx \frac{\sqrt{\pi}}{2} \, \frac{G(\x)}{\ntv(\x)} .
\end{align*}
The third line holds because $\nabla_i$ and $\F\HH \Dt_1 \F$ are both circulant on the periodic grid.
For the first approximation, we apply $\mathbb{E}[|v|] = (\sqrt{\pi} / 2) \sqrt{\mathbb{E}[|v|^2]}$~\cite{goodman2007speckle} for a Rayleigh-distributed variable $v$.
This is because that $\F\HH \Dt_1 \F$ is a circulant matrix with approximately Gaussian entries of variance $1 / n$~\cite{hu2026structured}.
Each entry of $\F\HH \Dt_1 \F \mathbf{g}_i$ is then approximately a circular complex Gaussian variable, and its magnitude is Rayleigh-distributed.
The second approximation uses $\|\mathbf{g}_1\|_2 + \|\mathbf{g}_2\|_2 \leq \sqrt{2} \, \|\nabla \conj{\boldsymbol{\Lambda} \x}\|_2$, with equality when the two axes carry the same difference energy.
The third approximation uses the fact that the phase steps of a signal under the $\boldsymbol{\Lambda}$ transform stay close to those of the original signal when the signal dimension is large.
\end{proof}
\begin{proposition} \label{prop:rho-two}
For $L = 2$, the TV ratio is approximated as
\begin{equation} \label{eq:rho-two}
    \rho_{\mathrm{TV}} \approx \frac{\sqrt{\pi}}{2} \, \frac{\sqrt{2}}{\ntv(\x)} .
\end{equation}
\end{proposition}
\begin{proof}
With $\mathbf{u} = \D_2 \F\HH \Dt_1 \F \conj{\boldsymbol{\Lambda} \D_2 \x}$,
\begin{align*}
    \rho_{\mathrm{TV}} &= \frac{\|\nabla \D_2^{-1} \F\HH \D_1^{-1} \F\HH \F \twin(\F\HH \F \D_1 \F \D_2 \x)\|_{1}}{\|\nabla \x\|_{1}} \ann{with \eqref{eq:srpr}, \eqref{eq:gtwin}, \eqref{eq:rho-tv}, $\G = \F \D_1 \F \D_2$} \\
    &= \frac{\|\nabla \mathbf{u}\|_{1}}{\|\nabla \x\|_{1}} \ann{with \eqref{eq:twin-fourier-dual}, \eqref{eq:twin-product}, $\F$ unitary} \\
    &\approx \frac{\sqrt{\pi}}{2} \, \frac{\sqrt{2 n} \, \|\nabla \mathbf{u}\|_{2}}{\|\nabla \x\|_{1}} \\
    &\approx \frac{\sqrt{\pi}}{2} \, \frac{\sqrt{2 n} \, \sqrt{4 n}}{\|\nabla \x\|_{1}} = \frac{\sqrt{\pi}}{2} \, \frac{\sqrt{2}}{\ntv(\x)} . \ann{with \eqref{eq:tv-g}}
\end{align*}
The entry $\mathbf{k}$ of $\mathbf{u}$ is $\d_2[\mathbf{k}] \sum_{\mathbf{l}} \mathbf{C}[\mathbf{k}, \mathbf{l}] \, \d_2[\mathbf{l}] \, \conj{(\boldsymbol{\Lambda} \x)[\mathbf{l}]}$, where $\mathbf{C} = \F\HH \Dt_1 \F$ is unitary.
The random signs $\d_2$ make the entries of $\mathbf{u}$ uncorrelated with zero mean, and approximately Gaussian by the central limit theorem.
Their variance is $\sum_{\mathbf{l}} |\mathbf{C}[\mathbf{k}, \mathbf{l}]|^2 = 1$ since $\x$ has unit magnitude, so $\mathbf{u}$ is white.
The entries are also approximately circular, $\mathbb{E}[\mathbf{u}[\mathbf{k}]^2] \approx 0$, as $\mathbf{C}[\mathbf{k}, \mathbf{l}]^2$ carries a linear phase of one full turn across the grid.
Each entry of $\nabla \mathbf{u}$ is thus approximately circular complex Gaussian with variance $2$.
The first approximation applies the Rayleigh property of Proposition~\ref{prop:rho-half} to the $2 n$ entries of $\nabla \mathbf{u}$, and the second uses $\|\nabla \mathbf{u}\|_2^2 \approx 2 n \cdot 2$.
\end{proof}
For $L \geq 2.5$, the generalized twin remains an uncorrelated Gaussian field of unit variance.
The same estimate therefore holds for every deeper model.

\section{Results}
\label{sec:results}
The results are reproducible with the source code publicly available.\footnote{\url{https://github.com/zhiyhucode/generalized-twin}}

\subsection{The TV Ratio Predicts Regularization Performance}
\label{sec:depth}
We compute the generalized twin and perform regularized reconstruction for the $256 \times 256$ cameraman image~\cite{tachella2025deepinverse} for $L$ from 0.5 to 2.5.
The grayscale image is remapped from $[0, 1]$ to $[0, \pi]$ as the phase of the unit-magnitude signal $\x$.
For each depth, we sample 10 sensing models with Rademacher masks and run one reconstruction for each from noiseless measurements at $\alpha = 0.5$.
The reconstruction first runs 5000 power iterations of the spectral method~\cite{maillard2022construction} using $\tau(y) = \max(1 - 1/y, -5)$ on the measurements normalized to unit mean, with early stopping to initialize.
It then runs 6000 iterations of PnP proximal gradient descent, using the gradient step of \eqref{eq:gradient} with $\mu = 2$.
The regularization step applies the pretrained DRUNet~\cite{zhang2021plug} as $D_\sigma$ to the real and imaginary parts at the noise level $\sigma_t = (1 - \epsilon)^{t}$ with $\epsilon = 5.7 \times 10^{-4}$.
We compute the TV ratio for the generalized twins and the peak signal-to-noise ratio (PSNR) of the reconstructed phase with peak $\pi$ after global phase alignment, and a reconstruction counts as successful above $30$ dB.

Fig.~\ref{fig:gallery} illustrates the ground truth as well as the generalized twins and the corresponding reconstructions at the median accuracy.
For $L=0.5$, the TV ratio is $1$ and all reconstructions stagnate at false solutions.
For $L=1$, the TV ratio increases to $10.7$ and all reconstructions converge to the ground truth.
For $L=1.5$, the TV ratio falls back to $1.9$ since the generalized twin becomes an all-pass random convolution, and only $2$ reconstructions out of $10$ succeed.
For $L \geq 2$, the generalized twin resembles noise with a TV ratio of $12.8$, and the reconstructions achieve slightly better accuracy than for $L = 1$.
These results indicate that the TV ratio faithfully reflects the tendency of the regularized reconstruction to avoid the ambiguities and converge to the ground truth.
\begin{figure}[t!]
    \centering
    \includegraphics[width=\linewidth]{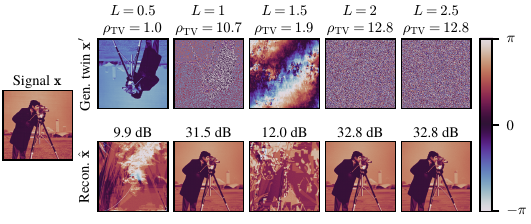}
    \caption{Generalized twin and regularized reconstruction for increasing depth $L$ at sampling ratio $0.5$. The TV ratio predicts the reconstruction accuracy measured by the PSNR in dB.}
    \label{fig:gallery}
\end{figure}

\subsection{Numerical Validation of Estimates of the TV Ratio}
\label{sec:check}
We compute the empirical TV ratios defined by \eqref{eq:rho-tv} for $L$ from $1$ to $2.5$ on the 68 images of the CBSD68 dataset~\cite{martin2001database}, with 10 sensing models per image, and compare them with the estimates of Propositions~\ref{prop:rho-one} to~\ref{prop:rho-two}.
The grayscale images are mapped to phases as in Section~\ref{sec:depth}.
The results are shown in Table~\ref{tab:forms}, where the maximum error across all depths is below 3\,\%, and the median error is below 1\,\%.
The errors come from the finite summation of random samples and from the assumptions that the signal has small phase steps and that the first-order differences per axis share the same energy.
\begin{table}[t]
    \centering
    \caption{Relative error of the closed-form estimates against the measured $\rho_{\mathrm{TV}}$ on CBSD68. All errors are below 3\,\%.}
    \label{tab:forms}
    \begin{tabular}{lcccc}
        \toprule
        $L$ & 1 & 1.5 & 2 & 2.5 \\
        \midrule[\heavyrulewidth]
        median & 0.6\,\% & 0.2\,\% & 0.1\,\% & 0.1\,\% \\
        90th percentile & 1.3\,\% & 0.7\,\% & 0.2\,\% & 0.2\,\% \\
        max & 2.6\,\% & 2.9\,\% & 0.4\,\% & 0.4\,\% \\
        \bottomrule
    \end{tabular}
\end{table}

\subsection{Sampling Ratio Scan}
\label{sec:sweep}
We compare the reconstruction performance with and without regularization for $L$ from $0.5$ to $2$ under different sampling ratios.
We use $\alpha = 0.1$ and $\alpha$ from $0.2$ to $4.0$ in steps of $0.2$.
We sample 50 sensing models for each depth and sampling ratio.
The signal and the regularized reconstruction algorithm are the same as in Section~\ref{sec:depth}, and the unregularized one removes the DRUNet step in each iteration.

Fig.~\ref{fig:sweep} plots the median PSNR of both reconstructions versus the sampling ratio per depth, with a violin for the PSNR distribution.
With $L = 0.5$ and $L = 1$, the unregularized reconstruction fails at every $\alpha$, while the generalized twin stays exact with oversampling.
In contrast, with $L \geq 1.5$, the generalized twin is no longer an exact ambiguity with oversampling, and unregularized reconstruction succeeds from $\alpha = 3.6$ for $L = 1.5$ and from $\alpha = 2.8$ for $L = 2$.
With regularization, the reconstruction performs similarly to the unregularized stagnation on average for $L = 0.5$, except that a few runs succeed from $\alpha = 1.6$.
In contrast, regularization enables stable and accurate reconstruction for $L = 1$ and $L = 2$ starting from $\alpha = 0.5$, with slightly better performance for $L = 2$ than for $L = 1$.
For $L = 1.5$, good reconstructions with regularization can be obtained for about half of the runs already at $\alpha = 0.6$, and the reconstructions approach the ground truth steadily for $\alpha > 3$.
The regularized PSNR saturates near $41$ dB, a ceiling set by the final noise level of the denoiser.
The generalized twin thus indicates whether an unregularized reconstruction can converge to the ground truth given sufficient measurements by whether it stays an exact ambiguity with these measurements.
The corresponding TV ratio further predicts the robustness and accuracy of the regularized reconstruction. 
\begin{figure}[t]
    \centering
    \includegraphics[width=\linewidth]{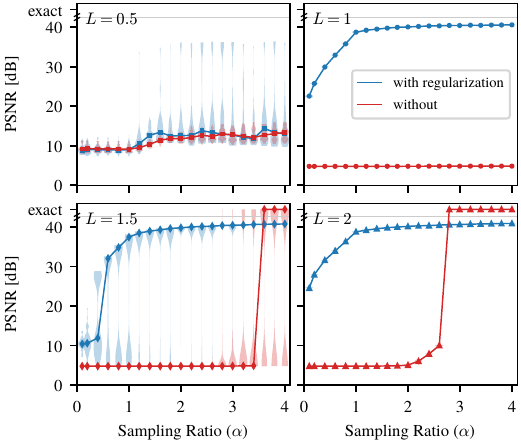}
    \caption{PSNR of the regularized and the unregularized reconstructions versus the sampling ratio, one panel per depth $L$. Whether the generalized twin stays an exact ambiguity indicates the unregularized performance. The corresponding TV ratio indicates the regularized performance.}
    \label{fig:sweep}
\end{figure}

\section{Conclusion}
\label{sec:conclusion}
We show that every invertible phase retrieval sensing operator keeps the twin image ambiguity from FPR as the generalized twin image, exactly with undersampling and, for a Fourier transform and an optional random mask, also with oversampling by zero-padding.
The structure of the generalized twin changes with the model structure.
We propose the TV ratio to quantify the variation of the generalized twin relative to the ground truth and show that the robustness and accuracy of regularized reconstruction scale with it.
In particular, the TV ratio is low for FPR, where regularization barely helps, and increases for $L=1$, where regularization enables stable and accurate reconstruction.
Since the TV ratio needs only the ground truth and the sensing model, and neither a reconstruction nor the denoiser, it predicts the benefit of regularization in advance and can guide the design of the sensing model.
The construction of the generalized twin applies to any ambiguity of FPR, and preliminary results for the circular translation lead to similar conclusions.
More generally, our framework links the symmetries of the forward model to the effectiveness of regularization.
The analysis of the ambiguities transformed from FPR can be extended to other modalities, such as ptychography~\cite{maiden2013ptychography} and Fourier ptychography.
This could help improve the convergence of reconstruction in many phase retrieval applications.

\section{Compliance with ethical standards}
\label{sec:ethics}
This is a theoretical and numerical-simulation study for which no ethical
approval was required.
\section{Acknowledgments}
\label{sec:acknowledgments}
This work was supported by the Swiss National Science Foundation (SNSF) under Grant 200020\_219356 and Grant PZ00P2\_216211 and by the European Research Council (ERC)
through the European Union’s Horizon 2020 (H2020), under Grant Agreement 101020573 FunLearn.

\bibliographystyle{IEEEbib}
\bibliography{refs}

\end{document}